\documentclass[conference]{IEEEtran}

\usepackage[nospace]{cite}

\usepackage{mathtools, amssymb}
    \allowdisplaybreaks
    \newtheorem{theorem}{Theorem}
    \newtheorem{lemma}[theorem]{Lemma}
    
    \newtheorem{proposition}[theorem]{Proposition}
    \def\ww{{\bar w}}
    \def\dd{\mathrm d}

    \def\le{\leqslant}
    \def\ge{\geqslant}
    \def\hom{\operatorname{hom}}
    \def\VVV#1#2#3{\mvhs{V}\limits^{#1\;#3}_{#2}}
    \def\NNN#1#2#3#4{\mvhs{N}\limits^{#1\;#3}_{#2\;#4}}
    \def\WWW#1#2#3#4#5{\mvhs{W}\limits^{#1\;#3\;#5}_{#2\;#4}}
    \def\mom#1#2#3{\vphantom|_{#1}\mkern-1mu |#2|_{#3}}
    \def\mvhs#1{\mathop{\vcenter{\hbox{\scalebox{-1.3}[1.3]{$\mathsf#1$}}}}}

\usepackage{tikz, pgfplots, pgfplotstable}
    \tikzset{
        every picture/.style={
            line cap=round, line join=round, baseline=-0.5ex,
        }
    }
    \pgfplotsset{
        compat=1.18,
        graph marker/.style={
            scatter, only marks, point meta=explicit symbolic,
            scatter/@pre marker code/.style={/tikz/mark=\pgfplotspointmeta#1},
            scatter/@post marker code/.style={},
        },
    }

\begin{document}

                                   \title
                 {Sidorenko-Inspired Pessimistic Estimation}

                                   \author
          {Yu-Ting Lin, Hsin-Po Wang (National Taiwan University)}

                                 \maketitle

\begin{abstract}\boldmath
    Recently, Abo Khamis et al.\ showed
    how to upper bound the size of a join of multiple tables,
    a problem essential to query optimization in database theory.
    They unified earlier works
    by the following information-theoretical framework.
    1. Let $(X_1, \dotsc, X_n)$ be a row
        selected from the join uniformly at random.
    2. The size of the join is now $\exp(H(X_1, \dotsc, X_n))$.
    3. To upper bound $H(X_1, \dotsc, X_n)$, break it into several
        \emph{local entropies}, such as $H(X_1)$, $H(X_2, X_3)$,
        and $H(X_4|X_5)$, using Shannon-type inequalities.
    4. Upper bound local entropies
        using statistics of the tables being joined.

    The statistics Abo Khamis et al.\ considered are
    the counts of graph homomorphisms from stars to the tables.
    In a follow-up work, we generalized stars to bi-stars.
    In this paper, we generalize bi-stars to caterpillars,
    an even larger class of graphs inspired by Sidorenko's conjecture.
    Simulations show that,
    while Abo Khamis et al.'s star bound overestimates the join size by $m$,
    our bi-star bound overestimates by about $m^{3/4}$,
    and this paper's new caterpillar bound overestimates by about $m^{3/5}$.
    These exponents are obtained by
    log-log regressions with R-square $> 0.98$.
    All homomorphisms are counted in time
    linear in the size of the tables being joined.
\end{abstract}

\section{Introduction}

    Relational databases are the de facto type of databases
    that have stood the test of time.
    Each relational database consists of \emph{relations}, or \emph{tables},
    with columns such as transaction ID, buyer ID, seller ID, and good ID.
    A buyer might use the same shipping address to purchase multiple goods,
    so the shipping address is stored in a different table to save space.

    Separating data into multiple tables means that,
    when those data are needed, we have to perform a \emph{table join}.
    Mathematically speaking, the join of a binary relation $R(A, B)$
    with another binary relation $S(B, C)$
    is the ternary relation
    \begin{equation}
        R(A, B) \land S(B, C) \coloneqq
        \{ (a, b, c) \text{ st } R(a, b) \text{ and } S(b, c) \}.
                                                               \label{Lambda}
    \end{equation}
    Here, $A$, $B$, and $C$ are column types,
    $R(A, B)$ means that the table $R$ contains columns $A$ and $B$,
    and $R(a, b)$ means that $(a, b) \in A \times B$ is a row of $R$.
    A more complicated example of a join query is
    \begin{equation}
        R(A, B) \land S(B, C) \land T(C, A) \coloneqq \Bigl\{ \substack{
            (a, b, c) \text{ st } R(a, b), \\
            S(b, c) \text{, and } T(c, a)
        } \Bigr\}.                                              \label{Delta}
    \end{equation}

    The long-standing challenge in database theory
    is to compute table joins as fast as possible.
    This problem is usually divided into two sub-goals:
    \begin{itemize}
        \item Find a \emph{pessimistic estimate}
            (an upper bound) on join size.
        \item Use the bound to make an informed plan such that
            the time it takes is $\le O(\text{the bound})$.
    \end{itemize}
    In the following, we review earlier works along this line
    using \eqref{Delta} as a running example.
    See also \cite{LGM15,HWW21,SuL20,PKB20,LBP21,KPN22}
    for benchmarks and surveys.

\subsection{Past Works}

    A trivial upper bound on the cardinality of \eqref{Delta} is
    \begin{equation}
        |\eqref{Delta}| \le |A| \cdot |B| \cdot |C|,           \label{triple}
    \end{equation}
    which corresponds to this execution plan.
    \begin{itemize}
        \item[] $>$ Enumerate all $a \in A$:
        \item[] ~~~$>$ Enumerate all $b \in B$:
        \item[] ~~~~~~$>$ Enumerate all $c \in C$:
        \item[] ~~~~~~~~~$>$ Collect $(a, b, c)$ st
                               $R(a, b)$, $S(b, c)$, $T(c, a)$.
    \end{itemize}
    This execution plan needs $O(\eqref{triple})$ time.
    The next bound is
    \begin{equation}
        |\eqref{Delta}| \le
        |R| \cdot |S|, |S| \cdot |T|, |T| \cdot |R|.            \label{2edge}
    \end{equation}
    The first upper bound corresponds to this execution plan.
    \begin{itemize}
        \item[] $>$ Enumerate all $(a, b) \in R$:
        \item[] ~~~$>$ Enumerate all $(b', c) \in S$:
        \item[] ~~~~~~$>$ Collect $(a, b, c)$ st $b = b'$ and $T(c, a)$.
    \end{itemize}
    This plan takes $O(|R| \cdot |S|)$ time to execute,
    again matching the upper bound.

    The next is the AGM (Atserias--Grohe--Marx) bound \cite{AGM13}
    \begin{equation}
        |\eqref{Delta}| \le
        |R(A, B)|^u \cdot |S(B, C)|^v \cdot |T(C, A)|^w,          \label{AGM}
    \end{equation}
    where $u + v, v + w, w + u \ge 1$.
    Note that
    $(u, v, w) = (0, 1, 1)$ gives back the second of \eqref{2edge} and
    $(u, v, w) = (1/2, 1/2, 1/2)$ gives a new, nontrivial inequality
    \begin{equation}
        |\eqref{Delta}| \le |R(A, B)|^{1/2} \cdot
        |S(B, C)|^{1/2} \cdot |T(C, A)|^{1/2}.                   \label{half}
    \end{equation}
    For how to generate execution plans matching the upper bounds,
    see \cite{NRR14,Vel14,NPR18}.

    Abo Khamis et al.\
    \cite{ANS16,ANS17} proposed
    \begin{equation}
        |\eqref{Delta}| \le
        |R(A, B)| \cdot \max_{a\in A} \deg_T(a),                \label{PANDA}
    \end{equation}
    where $\deg_T(a)$ is the number of $c \in C$ that satisfy $T(c, a)$.
    The execution plan matching \eqref{PANDA} can be found
    in those papers as well as in \cite{CHW22,CBS19,HHH21}.

    More recently, Abo Khamis, Nakos, Olteanu, and Suciu \cite{ANO24}
    introduced a major breakthrough unifying everything above
    into a single information-theoretical framework.
    We will detail their framework in Section~\ref{sec:four},
    but the key ingredients are the following statistics
    \[
        |{\deg_R(A)}|_p \coloneq
        \Bigl( \sum_{a\in A} \deg_R(a)^p \Bigr)\strut^{1/p},
    \]
    which they call the \emph{$p$-norm} of the \emph{degree sequence},
    as $\deg_R(A)$ is understood as the sequence
    $\deg_R(a_1)$, $\deg_R(a_2)$, $\dotsc$
    of degrees of elements of $A$.
    The $p$-norm is then used in the following upper bound
    \begin{equation}
        H(X) + p H(Y|X) \le \ln |{\deg_R(A)}|_p^p,                \label{p11}
    \end{equation}
    where $p \ge 0$ and $(X, Y) \in R(A, B)$ is a random row
    of $R$ that can follow any distribution.
    When $p = 0$, $|{\deg_R(A)}|_0^0$ becomes $|A|$,
    which is in \eqref{triple}.
    When $p = 1$, $|{\deg_R(A)}|_1^1 = |R|$,
    which is used in \eqref{AGM} and \eqref{PANDA}.
    When $p \to \infty$, the inequality degenerates into
    $H(Y|X) \le \ln \max_{a\in A} \deg_R(a)$,
    which is used in \eqref{PANDA}.
    See \cite{DSB22,DSB23,DSB25,ZMA25,MZA25} for follow-up works.

    Our previous work \cite{AMBI} extended this approach by introducing
    \emph{bivariate moments}\footnote{
    In \cite{AMBI}, we chose the exponents $p - 1$ and $q - 1$,
    rather than $p$ and $q$,
    so that \eqref{p1q} would not involve $(p + 1) H$ and $(q + 1) H$.
    In retrospect, this shifted indexing does not extend naturally
    to the caterpillar moments to be introduced later,
    and so we fall back to the more combinatorially natural indexing
    wherein $p$ means $p$ leaves attached, not $p \pm 1$.}
    \[
        \mom pRq \coloneqq \sum_{(a,b)\in R}
        \deg_R(a)^{p-1} \deg_R(b)^{q-1}.
    \]
    This adds an extra degree of freedom and can upper bound,
    for $p, q \ge 1$,
    \begin{equation}
        pH(Y|X) + I(X; Y) + q H(X|Y) \le \ln \mom pRq,            \label{p1q}
    \end{equation}
    which degenerates into the old bound \eqref{p11} when $q = 1$.

\subsection{This Work's New Contribution}

    In this paper, we introduce \emph{caterpillar moments}
    \begin{align*}
        \VVV pqr(R) & \coloneqq
        \sum_{\substack{R(a,b)\land R(c,b)}}
        \dd(a)^p \dd(b)^q \dd(c)^r,
        \\ \NNN pqrs(R) & \coloneqq
        \sum_{\substack{R(a,b)\land R(c,b)\land R(c,d)}}
        \dd(a)^p \dd(b)^q \dd(c)^r \dd(d)^s,
        \\ \WWW pqrst(R) & \coloneqq
        \sum_{\substack{R(a,b)\land R(c,b)\land\\R(c,d)\land R(e,d)}}
        \dd(a)^p \dd(b)^q \dd(c)^r \dd(d)^s \dd(e)^t,
    \end{align*}
    where $\dd$ is a shorthand for $\deg_R$.
    Caterpillar moments are generalizations of bivariate moments
    and $p$-norms as
    \begin{align*}
        \mom{p}{R}{1} & = |{\deg_R(A)}|_p^p,
        & \VVV{p}{q}{0}(R) & = \mom{p+1}{R}{q+2},
        \\ \NNN{p}{q}{r}{0}(R) & = \VVV{p}{q}{r+1}(R),
        & \WWW{p}{q}{r}{s}{0}(R) & = \NNN{p}{q}{r}{s+1}(R).
    \end{align*}
    With the new statistics we generalize \eqref{p11} and \eqref{p1q} into
    \small
    \begin{gather*}
        (p + r + 1) H(Y|X) + I + (q + 2) H(X|Y)
        \le \ln \VVV pqr(R),
        \\ (p + r + 2) H(Y|X) + I + (q + s + 2) H(X|Y)
        \le \ln \NNN pqrs(R),
        \\ (p + r + t + 2) H(Y|X) + I + (q + s + 3) H(X|Y)
        \le \ln \WWW pqrst(R),
    \end{gather*}
    \normalsize
    where $I$ stands for $I(X; Y)$
    and $p, q, r, s, t \ge 0$.
    These moments are inspired by Sidorenko's conjecture \cite{Sze15},
    which involves comparing the number of graph homomorphisms.
    We found that graphs admitting fewer homomorphisms
    in the Sidorenko sense yield better upper bounds on join size.

    This paper is organized as follows.
    Section~\ref{sec:four} reviews the information-theoretical framework.
    Section~\ref{sec:star} reviews \eqref{p11}.
    Section~\ref{sec:bistar} reviews \eqref{p1q} from our previous work.
    In Section~\ref{sec:caterpillar}, we explain the main contribution:
    the caterpillar bounds.

\section{The Information-Theoretical Framework}              \label{sec:four}

    The works mentioned above are all based on an information-theoretical
    observation: the size of a set is the exponential
    of the entropy of a uniformly-sampled element of the set.
    So let $(X_1, X_2, \ldots, X_n) $ be a random row selected uniformly
    from a join query, where $X_i$ is the entry in the $i$th column.
    Then
    \begin{equation}
        H(X_1, X_2, \ldots, X_n) = \ln(\text{table size})      \label{lnsize}
    \end{equation}
    Therefore, finding an upper bound on the table size
    is equivalent to bounding the join entropy from above.
    
    We need machinery that generates upper bounds on
    \eqref{lnsize} using simpler local entropies
    of the form $H(X_i)$, $H(X_i, X_j)$, and $H(X_j|X_i)$;
    this is presented in Subsection A.
    Next, we need a way to upper bound these local entropies using numbers
    that can be (easily) computed from the individual tables being joined;
    this is presented in Subsection B.
    Chaining the two subsections will result in something like this:
    \[
        H(X, Y, Z)
        \stackrel{\text{\S A}}\le H(X) + H(Y, Z)
        \stackrel{\text{\S B}}\le \ln|A| + \ln|S|
    \]
    Finally, a linear program finds the best bound
    that can be obtained this way, as detailed in Subsection C.

\subsection{Shannon-Type Inequalities}

    Shannon-type inequalities provide a systematic and powerful way
    to break \eqref{lnsize} into local entropies.
    It is quite common to classify them into four subtypes:
    nonnegativity, monotonicity, subadditivity, and submodularity
    \begin{gather}
        H(X) \ge 0,                                              \label{nonn}
        \\ H(X, Y) \ge H(X),                                     \label{mono}
        \\ H(X) + H(Y) \ge H(X, Y),                              \label{suba}
        \\ H(X, Y) + H(X, Z) \ge H(X, Y, Z) + H(X).              \label{subm}
    \end{gather}
    Proofs of them are omitted for brevity.

    Any substitution and combination of these inequalities
    is also regarded as Shannon-type.
    For instance,
    \[ H(X, Y) \ge H(X) \ge 0 \]
    is just the sum of \eqref{nonn} and \eqref{mono}.
    Similarly,
    \[
        H(X) + H(Y) + H(Z) \ge H(X, Y) + H(Z) \ge H(X, Y, Z)
    \]
    is just the sum of \eqref{suba} and \eqref{suba} with $(X, Y)$
    replaced by $((X, Y), Z)$, and therefore it is of Shannon-type.
    Moreover,
    \[ H(X, Y, Y') + H(X, Z, Z') \ge H(X, Y, Y', Z, Z') + H(X) \]
    which is \eqref{subm} with $(Y, Z)$ replaced by $((Y, Y'), (Z, Z'))$.
    It can also be proved by combining
    \begin{itemize}
        \item \eqref{subm} with $(Y, Z)$ replaced by $(Y, Z)$,
        \item \eqref{subm} with $(X, Y)$ replaced by $((X, Y), Y')$,
        \item \eqref{subm} with $(X, Z)$ replaced by $((X, Z), Z')$, and
        \item \eqref{subm} with $(X, Y, Z)$ replaced by
            $((X, Y, Z), Y', Z')$.
    \end{itemize}
    Since we will be using a linear programming solver at the last step,
    inequalities that are derivable from others
    by linear combinations are considered redundant.
    The following lemma gives a minimal generating set
    of Shannon-type inequalities.

    \begin{proposition} [generators]                          \label{pro:gen}
        All Shannon-type inequalities are
        linear combinations of the following two families:
        (a) monotonicity of adding one variable
        \[ H(X_1, X_2, \dotsc, X_n , Y) \ge H(X_1, X_2, \dotsc, X_n) \]
        and (b) submodularity with single $Y$ and single $Z$
        \begin{align*}
            & H(X_1, \dotsc, X_n , Y) + H(X_1, \dotsc, X_n , Z)
            \\ &\ge H(X_1, \dotsc, X_n, Y, Z) + H(X_1, \dotsc, X_n).
        \end{align*}
    \end{proposition}

    When there are $m$ variables,
    submodularity alone gives $4^m$  inequalities,
    while (a) gives $m 2^{m-1}$ and (b) gives $ \binom m2 2^{m-2}$.
    The proof of Proposition~\ref{pro:gen} is omitted because
    its absence does not affect the validity of our bounds.
    It just means that we use weaker inequalities than there are available.

\subsection{Table Statistics for Local Entropies}

    We now describe how to upper bound local entropy terms
    such as $H(X)$, $H(X, Y)$, and $ H(Y|X)$.
    For any random pair $(X, Y) \in R(A, B)$,
    \begin{align}
        H(X) & \le \ln |A|,                                       \label{p=0}
        \\ H(X, Y) & \le \ln |R(A, B)|,                           \label{p=1}
        \\ H(Y|X) & \le \ln \max_{a\in A} \deg_R(a).             \label{p=oo}
    \end{align}
    All three inequalities are straightforward consequences of
    the fact that entropy is maximized by the uniform distribution.

    Abo Khamis et al.\ \cite{ANO24} unified \eqref{p=0}--\eqref{p=oo}
    as \eqref{p11} with a free parameter $p \ge 0$.
    Now \eqref{p=0}, \eqref{p=1}, and \eqref{p=oo} are special cases
    of \eqref{p11} at $p = 0$, $1$, and $\infty$, respectively.
    Our previous work \cite{AMBI} further generalized \eqref{p11}
    to \eqref{p1q}, making \eqref{p11} a special case at $q = 1$.
    In this work, we provide new upper bounds
    using caterpillar moments.

    A crucial feature of these statistics
    is their computational accessibility.
    All quantities involved—domain sizes, $p$-norms of degree sequences,
    bivariate moments, and caterpillar moments---can
    be computed in time linear in $|R|$.
    This property is essential for practical query optimization.

\subsection{Linear Programming}

    As established earlier, we want to bound
    $H(X_1, \dotsc, X_n) = \ln(\text{table size})$
    from above using Shannon-type inequalities.
    But there are infinitely many bounds,
    so we need to solve an optimization problem:
    \begin{align*}
        \text{minimize}\quad
        & \sum \text{coef} \cdot \text{table-statistics} \\
        \text{subject to}\quad 
        & \text{coef must form valid Shannon-type ineq} \\
        & \text{that upper bounds } H(X_1, \dotsc, X_n).
    \end{align*}
    Any admissible value of this linear program
    is an upper bound on the join cardinality,
    so the minimum thereof should be competitively tight.
    By duality, this is equivalent to
    \begin{align*}
        \text{maximize}\quad & H(X_1, \dotsc, X_n) \\
        \text{subject to}\quad 
        & \text{inequalities from Section 2.A,} \\
        & \text{inequalities from Section 2.B.}
    \end{align*}

    The linear program unifies a wide spectrum of pessimistic cardinality 
    estimation techniques---from AGM to caterpillar.
    Our contribution here is to introduce new table statistics
    to equip the linear program with stronger constraints.

\section{Counting Stars}                                     \label{sec:star}

    In the information-theoretical framework of Section~\ref{sec:four}, the
    entropy of a join is decomposed into several local entropies.
    To upper bound those local entropies, Abo Khamis et al. 
    \cite{ANO24} use $p$-norms of degree sequences.
    We review their approach in this section.

\subsection{The Combinatorial Interpretation}

    As mentioned before,
    Abo Khamis et al. \cite{ANO24} introduced
    \[ |{\deg_R(A)}|_p^p = \sum_{a\in A} \deg_R(a)^p. \]
    When $p$ is an integer (it need not be), $|{\deg_R(A)}|_p^p$
    counts the number of graph homomorphisms from $K_{1,p}$ to $R$.
    Here, $K_{1,p}$ consists of a root and $p$ leaves,
    and it is usually called a \emph{star} in graph theory.
    To determine a homomorphism, first map the root to some $a \in A$,
    and select the image of each leaf, which
    can be any of the $\deg_R(a)$ neighbors of $a$ in $R$.
    Therefore, there are $\deg_R(a)^p$ ways to map the leaves.
    Summing over all choices of $a$ gives $|{\deg_R(A)}|_p^p$ choices.

\subsection{From Stars to Local Entropy}

    The next step is to relate
    the count of homomorphisms to local entropy.
    Let $(X, Y) \in R(A, B)$ be a random row
    that can follow any distribution.
    We let $X_0$ be an iid copy of $X$,
    and let $Y_1, \dotsc, Y_p$ be independently generated
    by the conditional probabilities $P_{Y|X}(b|X_0)$.
    This way, $(X_0, Y_i)$ for each $i$
    follows the same join distribution as $(X, Y)$.
    Hence
    \begin{align*}
        H(X_0, Y_1, \dotsc, Y_p)
        & = H(X) + \sum_{i=1}^p H(Y_i|X) \\
        & = H(X) + p H(Y|X).
    \end{align*}
    Note that $H(X_0, Y_1, \dotsc, Y_p)$ is the entropy of a random star,
    and so it is less than or equal to $\ln(\#\text{stars})$.
    Hence $\ln \mom{p}{R}{1} \ge H(X) + p H(Y|X)$,
    proving \eqref{p11} for the integer case.

    Now an interesting extension is that $p$ need not be an integer.
    
    \begin{lemma}                                             \label{lem:p11}
        (Lemma~4.1 of \cite{ANO24}) Let $p \ge 0$.
        Let $R(A, B)$ be a relation with column types $A$ and $B$.
        Let $(X, Y) \in R(A, B)$ be any random pair.
        Then \eqref{p11} holds.
    \end{lemma}
    
    The proof of Lemma~\ref{lem:p11} is omitted.
    Stars constitute the first level of the statistics hierarchy:
    they capture one-hop neighborhoods and serve as a conceptual baseline
    for stronger statistics to be introduced later,
    namely bi-stars (two-hop) and caterpillars (multi-hop).

\section{Counting Bi-stars}                                \label{sec:bistar}

    Our contribution in the last paper \cite{AMBI} was to introduce
    \[
        \mom{p}{R(A, B)}{q} \coloneqq
        \sum_{(a,b)\in R} \deg_R(a)^{p-1} \deg_R(b)^{q-1},
    \]
    which specialize to the old $p$-norms when $q = 1$:
    \[
        \mom{p}{R}{1}
        = \sum_{(a,b)\in R}\! \deg_R(a)^{p-1}
        = \sum_{a\in A} \deg_R(a)^p
        = |{\deg_R(A)}|_p^p
    \]
    We then generalize \eqref{p=0}--\eqref{p=oo} and \eqref{p11} into
    \eqref{p1q} for all $p, q \ge 1$,
    making \eqref{p11} a special case of \eqref{p1q} at $q = 1$.

\subsection{The Combinatorial Interpretation}

    When $p$ and $q$ are positive integers, $\mom{p}{R(A, B)}{q}$
    counts the number of graph homomorphisms from a \emph{bi-star} to $R$.
    Here, a $(p, q)$-bi-star, usually denoted by $S_{p,q}$,
    is formed by two stars $K_{1,p}$ and $K_{1,q}$
    with an edge connecting their roots.
    To determine a homomorphism from $S_{p-1,q-1}$ to $R$,
    we map the roots to some edge $(a, b) \in R$;
    then there are $\deg_R(a)^{p-1}$ ways
    to map $p - 1$ leaves to the $B$-side times
    $\deg_R(b)^{q-1}$ ways to map $q - 1$ leaves to the $A$-side.

    \subsection{From Bi-stars to Local Entropy}

    Let $(X, Y) \in R(A,B)$ be a random row that can follow any distribution.
    We let $(X_0, Y_0)$ be an iid copy of $(X, Y)$,
    let $Y_1, \dotsc, Y_{p-1}$ be independently generated
    by the conditional probabilities $P_{Y|X}(b|X_0)$,
    and let $X_1, \dotsc, X_{q-1}$ be independently generated
    by the conditional probabilities $P_{X|Y}(a|Y_0)$.
    This way, each $(X_0, Y_i)$ and each $(X_j, Y_0)$
    follow the same join distribution as $(X, Y)$.
    Consequently,
    \begin{align*}
        \kern2em&\kern-2em
        H(X_0, Y_0, Y_1, \dotsc, Y_{p-1}, X_1, \dotsc, X_{q-1})
        \\ & = H(X_0, Y_0)
            + \sum_{i=1}^{p-1} H(Y_i|X_0)
            + \sum_{j=1}^{q-1} H(X_j|Y_0)
        \\ & = p H(Y|X) + I(X; Y) + q H(X|Y).
    \end{align*}
    Note that the left-hand side is the entropy of a random bi-star,
    which is less than or equal to $\ln(\#\text{bi-stars})$.
    Hence \eqref{p1q} holds for positive integers $p$ and $q$.
    The following lemma extends \eqref{p1q} to all real $p, q \ge 1$.    

    \begin{lemma}                                             \label{lem:p1q}
        (Theorem~3 of \cite{AMBI})
        Let $p, q \ge 1$.  Let $R(A, B)$ be a relation with column types $A$
        and $B$.  Let $(X, Y) \in R(A, B)$ be any random pair.  Then $p
        H(Y|X) + I(X; Y) + q H(X|Y) \le \ln \mom{p}{R(A, B)}{q}$, i.e.,
        \eqref{p1q} holds.
    \end{lemma}
    
    The proof of Lemma~\ref{lem:p1q} is omitted here.
    Besides, we also proved that bi-star moments
    are log-convex in $(p, q)$.

    \begin{proposition}                                    \label{pro:holder}
        (Theorem~5 of \cite{AMBI})
        Let $p, q, r, s \ge 1$ and $0 < w < 1$.  Let $R(A, B)$ be a relation.
        Then
        \[
            \mom pRq^w \cdot \mom{r}{R}{s}^{1-w} \ge
            \mom{wp+(1-w)r}{R}{wq+(1-w)s},
        \]
        i.e., $\ln \mom pRq$ and $\ln |{\deg_R(a)}|_p^p$
        are convex in $(p, q)$.
    \end{proposition}

    This convexity has an important conceptual consequence:
    For any upper bound in the form of a linear combination of
    several bi-star bounds, we can greedily tighten it
    by combining two $\mom{}R{}$ into one $\mom{}R{}$.
    This explains why introducing bi-stars brings tighter bounds:
    By treating $|{\deg_R(A)}|_p^p$ and $|{\deg_R(B)}|_q^q$
    as $\mom{p}{R}{1}$ and $\mom{1}{R}{q}$, respectively,
    their convex combination is tighter,
    but this is only available in our previous work \cite{AMBI}. 
    
    Bi-stars form the second level of the statistics hierarchy:
    they capture two-hop neighborhoods.
    Drawing inspiration from Sidorenko's conjecture,
    we now proceed to the third level---the multi-hop statistics.
    
\section{Counting Caterpillars}                       \label{sec:caterpillar}

    In the previous work \cite{AMBI}, we showed that
    bi-star moments provide tighter upper bounds for the linear program
    without changing the information-theoretical framework.
    Naturally one asks if even tighter statistics exist.

    In this section, we explain how Sidorenko inspired us.
    In particular, our new bounds do not depend on
    the correctness of Sidorenko's conjecture.
    Instead, the conjecture guides us to compute
    statistics that \emph{conjecturally} yield tighter bounds.

\subsection{From Sidorenko's Conjecture to Caterpillars}

    Sidorenko's conjecture asserts that for every bipartite graph $H$,
    its homomorphism density into any host graph $G$
    is minimized by the random graph of the same edge density.
    In other words, the conjecture compares
    the number of graph homomorphisms from $H$ to $G$
    and the number of graph homomorphisms from $K_2$ to $G$,
    where $K_2$ is a single edge.

    As of today, the full conjecture remains open,
    but graphs $H$ known to satisfy it are called \emph{Sidorenko graphs},
    and it is an important research direction to identify more such graphs.
    So far, the largest class of Sidorenko graphs
    is provided by Szegedy \cite{Sze15} using entropy methods,
    which involves assigning probability distributions to homomorphisms
    to build larger Sidorenko graphs from smaller ones.

    Informally speaking, our viewpoint of the conjecture is as follows.
    To study homomorphisms from $K_4$ to $G$,
    understanding homomorphisms from $K_3$ is more helpful than from $K_2$.
    In other words, we can study $\hom(H, G)$ by ``Taylor-expanding'' $H$.
    A better approximation of $H$ should yield better control.
    
    Returning to the context of pessimistic estimation, one can see that
    $\VVV pqr(R)$ naturally extends bivariate moments
    from summing over edges to summing over $3$-vertex paths,
    while bivariate moments are already an extension of $p$-norms
    from summing over vertices to summing over edges.
    Likewise, $\NNN pqrs(R)$ extends $\VVV pqr$
    from $3$-vertex paths to $4$-vertex paths,
    and $\WWW pqrst(R)$ extends them to $5$-vertex paths.

    Why longer paths?
    The summand $\deg_R(a)^p \deg_R(b)^q$ is a special case of
    $\deg_R(a)^{p/2} \* \deg_R(b)^q \* \deg_R(c)^{r/2}$ when $c = a$.
    So the former is the ``diagonal term'' of the latter.
    By a Cauchy--Schwarz-type argument,
    uncorrelated products are less than correlated products,
    and hence $\VVV pqr$ should provide tighter bounds.
    Similarly, when summing over longer paths,
    especially $\NNN{p}{0}{0}{s}$ and $\WWW{p}{0}{0}{0}{t}$,
    the endpoints are further apart, and so the correlation is even weaker.
    
\subsection{The Combinatorial Interpretation}

    When $p$, $q$, $r$, $s$, and $t$ are positive integers,
    $\VVV pqr(R)$, $\NNN pqrs(R)$, and
    $\WWW pqrst(R)$ count the number of graph homomorphisms
    from caterpillars to $R$.
    A caterpillar, in graph theory, is a tree
    that reduces to a path when all leaves are removed.
    It generalizes stars (the remaining path is a single vertex)
    and bi-stars (the remaining path is an edge).
    We choose the notation so that
    $\mvhs{V}$, $\mvhs{N}$, and $\mvhs{W}$
    resemble the paths of $3$, $4$, and $5$ vertices, respectively,
    and $p$, $q$, $r$, $s$, and $t$ correspond to
    the number of leaves attached to each vertex.

    In particular, $\VVV pqr$ counts homomorphisms
    from a 3-vertex path with $p$ leaves attached to one endpoint,
    $q$ leaves attached to the middle vertex, and
    $r$ leaves attached to the other endpoint.
    This is because the $p$ leaves each have $\deg_R(a)$ choices,
    the $q$ leaves each have $\deg_R(b)$ choices, and
    the $r$ leaves each have $\deg_R(c)$ choices.
    Similarly, $\NNN pqrs$ counts homomorphisms
    from a 4-vertex path with $p$, $q$, $r$, and $s$ leaves
    attached to the four vertices, respectively.
    Finally, $\WWW pqrst$ counts homomorphisms
    from a path with $p$, $q$, $r$, $s$, and $t$ leaves attached.
    By the following theorem these quantities are easy to compute.

    \begin{theorem}
        For a fixed $(p, q, r, s, t)$, $\VVV pqr(R)$,
        $\NNN pqrs(R)$, and $\WWW pqrst(R)$
        can be computed in time linear in the size of $R$.
    \end{theorem}

    \begin{IEEEproof}
        In linear time we can compute $\deg_R(c)^r$ for each $c \in A$.
        In another linear time we can then compute
        $\sum_{c:R(c,b)} \deg_R(b)^q \deg_R(c)^r$ for each $b \in B$.
        Afterward we compute
        $\sum_{b:R(a,b)} \sum_{c:R(c,b)} \deg_R(a)^p \deg_R(b)^q \deg_R(c)^r$
        for each $a \in A$.
        Now $\VVV pqr(R)$ is just the sum over all $a \in A$.
        For the other two caterpillar moments, similar processes apply.
    \end{IEEEproof}

\subsection{From Caterpillars to Local Entropy}

    We now explain why caterpillar moments
    upper bound local entropies.
    The idea is similar to those of stars and bi-stars.
    We let $(X, Y)$ follow any distribution on $R(A, B)$.
    Let $(X_0, Y_0, X'_0)$ be a random path
    such that $(X_0, Y_0)$ and $(X'_0, Y_0)$
    follow the same distribution as $(X, Y)$.
    Then we generate $Y_1, \dotsc, Y_p$ from $X_0$,
    $X_1, \dotsc, X_q$ from $Y_0$, and
    $Y'_1, \dotsc, Y'_r$ from $X'_0$.
    This way, each $H(Y_i | X_0)$ is $H(Y|X)$,
    each $H(X_j | Y_0)$ is $H(X|Y)$, and
    each $H(Y'_k | X'_0)$ is $H(Y|X)$.
    In total, we have entropy terms
    \[ (p + r + 1) H(Y|X) + I(X; Y) + (q + 2) H(X|Y), \]
    which is less than or equal to the logarithm of
    the number of caterpillars, $\VVV pqr(R)$.
    With the same reasoning applied to $\NNN pqrs$ and $\WWW pqrst$,
    we have the following theorem.

    \begin{theorem}                                           \label{thm:int}
        With integral $p, q, r, s, t \ge 0$ and $I = I(X; Y)$,
        {\small
        \begin{gather*}
            (p + r + 1) H(Y|X) + I + (q + 2) H(X|Y)
                \le \ln \VVV pqr(R),
            \\ (p + r + 2) H(Y|X) + I + (q + s + 2) H(X|Y)
                \le \ln \NNN pqrs(R),
            \\ (p + r + t + 2) H(Y|X) + I + (q + s + 3) H(X|Y)
                \le \ln \WWW pqrst(R).
        \end{gather*}}
    \end{theorem}

\begin{tikzpicture} [overlay]
    \path foreach \i in {0, ..., 2} {(\i*120 + 90: 2pt) coordinate (c\i)};
    \path foreach \j in {0, ..., 3} {(22.5 + \j*90: 2pt) coordinate (d\j)};
    \path foreach \k in {0, ..., 4} {(\k*72: 2pt) coordinate (\k)};
\end{tikzpicture}
\def\setpic#1#2{\pgfdeclareplotmark{#2}{\pgftext{\copy#1}}}
\def\deficon#1#2{
    \newbox\iconbox\setbox\iconbox=
    \hbox{\tikz \draw [line width=0.2pt, blue] #2;}
    \expandafter\setpic\expandafter{\the\iconbox}{#1b}
    \newbox\iconbox\setbox\iconbox=
    \hbox{\tikz \draw [line width=0.2pt, cyan] #2;}
    \expandafter\setpic\expandafter{\the\iconbox}{#1c}
    \newbox\iconbox\setbox\iconbox=
    \hbox{\tikz \draw [line width=0.2pt, green] #2;}
    \expandafter\setpic\expandafter{\the\iconbox}{#1g}
    \newbox\iconbox\setbox\iconbox=
    \hbox{\tikz \draw [line width=0.2pt, yellow] #2;}
    \expandafter\setpic\expandafter{\the\iconbox}{#1y}
    \newbox\iconbox\setbox\iconbox=
    \hbox{\tikz \draw [line width=0.2pt, red] #2;}
    \expandafter\setpic\expandafter{\the\iconbox}{#1r}
}
\deficon{path3}{(c2)--(c0)--(c1)}
\deficon{K3}{(c2)--(c0)--(c1)--cycle}
\deficon{claw}{(d0)--(d1) (d0)--(d2) (d0)--(d3)}
\deficon{path4}{(d0)--(d1)--(d2)--(d3)}
\deficon{pan3}{(d3)--(d1)--(d2)--(d3)--(d0)}
\deficon{cycle4}{(d0)--(d1)--(d2)--(d3)--(d0)}
\deficon{fan2}{(d0)--(d1)--(d2)--(d3)--(d0) (d0)--(d2)}
\deficon{K4}{(d0)--(d1)--(d2)--(d3)--(d0)--(d2) (d1)--(d3)}
\deficon{K14}{(0)--(1) (0)--(2) (0)--(3) (0)--(4)}
\deficon{chair}{(1)--(2)--(3)--(4) (0)--(2)}
\deficon{path5}{(3)--(4)--(0)--(1)--(2)}
\deficon{cricket}{(2)--(3)--(0)--(2) (1)--(0)--(4)}
\deficon{pan4}{(3)--(4)--(0)--(1)--(2) (1)--(3)}
\deficon{bull}{(3)--(4)--(0)--(1)--(2) (1)--(4)}
\deficon{pan4c}{(3)--(4)--(0)--(1)--(2) (0)--(3)}
\deficon{cycle5}{(0)--(1)--(2)--(3)--(4)--(0)}
\deficon{dart}{(3)--(4)--(0)--(1)--(2) (1)--(4)--(2)}
\deficon{K23}{(3)--(4)--(0)--(1)--(2) (1)--(3) (2)--(4)}
\deficon{butterfly}{(3)--(4)--(0)--(1)--(2) (2)--(0)--(3)}
\deficon{house}{(0)--(1)--(2)--(3)--(4)--(0) (1)--(4)}
\deficon{kite}{(3)--(4)--(0)--(1)--(2) (0)--(2)--(4)}
\deficon{K3u2K1c}{(3)--(4)--(0)--(1)--(2)--(4)--(3)--(1)}
\deficon{fan3}{(0)--(1)--(2)--(3)--(4)--(0) (2)--(0)--(3)}
\deficon{clawuK1c}{(3)--(4)--(0)--(1)--(2) (0)--(3)--(1)--(4)}
\deficon{P2uP3c}{(0)--(1)--(2)--(3)--(4)--(0) (1)--(3) (2)--(4)}
\deficon{P3u2K1c}{(0)--(1)--(2)--(3)--(4)--(0)(2)--(4)--(1)--(3)}
\deficon{wheel4}{(2)--(0)--(1)--(2)--(3)--(4)--(0)--(2)(1)--(4)}
\deficon{K5_e}{(1)--(2)--(3)--(4)--(0)--(1)--(3)--(0)--(2)--(4)}
\deficon{K5}{(1)--(2)--(3)--(4)--(0)--(1)--(3)--(0)--(2)--(4)--(1)}

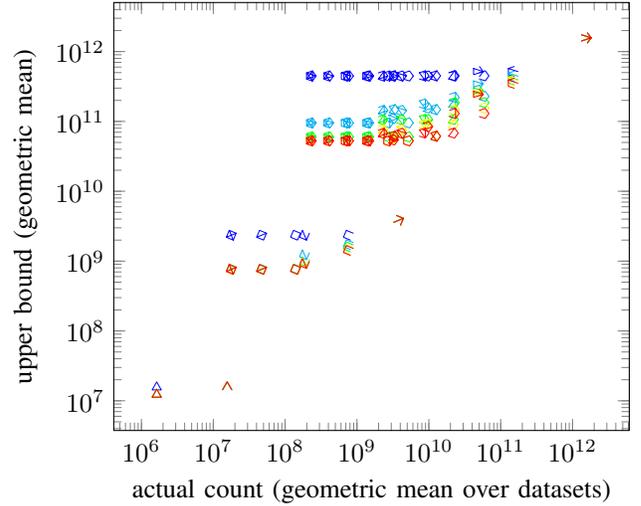
\begin{figure}
    \begin{tikzpicture}
        \begin{axis} [
            xmode=log, xlabel={actual count (geometric mean over datasets)},
            ymode=log, ylabel={upper bound (geometric mean)},
        ]
            \addplot [graph marker=b] table [x=true, y=star, meta=shape]
                {fig54/_average.csv};
            \addplot [graph marker=c] table [x=true, y=bistar, meta=shape]
                {fig54/_average.csv};
            \addplot [graph marker=g] table [x=true, y=vvv, meta=shape]
                {fig54/_average.csv};
            \addplot [graph marker=y] table [x=true, y=nnn, meta=shape]
                {fig54/_average.csv};
            \addplot [graph marker=r] table [x=true, y=www, meta=shape]
                {fig54/_average.csv};
        \end{axis}
    \end{tikzpicture}
    \caption{
        The horizontal axis is the actual number of homomorphisms;
        the vertical axis is the upper bound;
        both axes are geometrically averaged over
        the $39$ datasets we tested.
        The icons represent $H$.
        The colors represent the five methods---longer
        wavelengths mean longer caterpillars.
    }                                                         \label{fig:sim}
\end{figure}

\section{Simulation}

    Note that, when $H$ is a path with $3$ vertices,
    each homomorphism from $H$ to $G$ corresponds to a row of \eqref{Lambda}
    where $R = S = T = E(G)$ and $A = B = C = V(G)$.
    Likewise, when $H$ is $K_3$,
    each homomorphism from $H$ to $G$ corresponds to a row of \eqref{Delta}.
    For bigger $H$, graph homomorphisms correspond to more complicated joins.

    So similar to \cite{AMBI}, we download datasets from
    SNAP\footnote{\texttt{https://snap.stanford.edu/data/}}
    and use them as graph $G$.
    We then count the numbers of homomorphisms from $H$ to $G$
    for all small $H$ with at most $5$ vertices.
    
    We then compare the following five methods.
    \begin{itemize}
        \item Abo Khamis et al.'s star-based bound \cite{ANO24}.
        \item The last item plus $\mom pRq$
            for $p, q \in \{2, 3, 4, 5\}$ \cite{AMBI}.
        \item The last item plus $\VVV p0p$
            for $p \in \{1, 2, 3\}$.
        \item The last item plus $\NNN p00p$
            for $p \in \{1, 2, 3\}$.
        \item The last item plus $\WWW p000p$
            for $p \in \{1, 2, 3\}$.
    \end{itemize}
    The results are presented in Figure~\ref{fig:sim}.

\section{Conclusion}

    In this paper, we progressively generalize star
    to bi-stars and then to caterpillars.
    Each step refines the previous bound by what
    Sidorenko's conjecture suggests to be tighter statistics.
    Empirically, when the star bound overestimates the join size by $m$,
    our bi-star bound overestimates by about $m^{3/4}$,
    and this paper's new caterpillar bound overestimates by about $m^{3/5}$.

\bibliographystyle{IEEEtran}
\bibliography{SidorenQL-31}

\appendices

\section{Real Caterpillar Moments and Convexity}

    Theorem~\ref{thm:int} in Section~\ref{sec:caterpillar} is limted to
    integer parameters because caterpillars have whole numbers of leaves.
    Nevertheless, the entropy bounds themselves extend to
    nonnegative real numbers, just as bi-star bounds
    Lemma~\ref{lem:p1q} assume real parameters.
    We record this here for completeness.

    \begin{theorem}                                          \label{thm:real}
        Let $p, q, r, s, t \ge 0$.
        Let $R(A, B)$ be a relation with column types $A$ and $B$.
        Let $(X, Y) \in R(A, B)$ be any random pair.
        Then
        \begin{gather*}
            (p + r + 1) H(Y|X) + I + (q + 2) H(X|Y)
                \le \ln \VVV pqr(R),
            \\ (p + r + 2) H(Y|X) + I + (q + s + 2) H(X|Y)
                \le \ln \NNN pqrs(R),
            \\ (p{+}r{+}t{+}2) H(Y|X) + I + (q{+}s{+}3) H(X|Y)
            \le \ln \WWW pqrst(R),
        \end{gather*}
        where $I$ is the mutual information $I(X; Y)$.
    \end{theorem}

    \begin{IEEEproof}
        Write $\dd$ for $\deg_R$.
        For $\VVV pqr(R)$, sample a random path $(X_0, Y_0, X_1)$ by
        first sampling $Y_0 \sim P_Y$ and then sampling
        $X_0$ and $X_1$ independently from $P_{X|Y}(a|Y_0)$.
        Denote its pmf by
        \[ Q(a, b, c) \coloneqq P_Y(b) P_{X|Y}(a|b) P_{X|Y}(c|b). \]
        Then
        \begin{multline*}
            H(Y|X) + I(X; Y) + 2 H(X|Y)
            \\ = H(X_0, Y_0, X_1)
                = \sum_{a,b,c} Q(a, b, c) \ln \frac{1}{Q(a, b, c)}.
        \end{multline*}
        Since $Y$ is supported on at most $\dd(a)$ points given $X = a$,
        Jensen's inequality gives
        \[
            p H(Y|X)
            \le p \sum_{a\in A} P_X(a) \ln \dd(a)
            = \sum_{a,b,c} Q(a, b, c) \ln \dd(a)^p.
        \]
        By symmetry,
        \begin{align*}
            q H(X|Y)
            & \le \sum_{a,b,c} Q(a, b, c) \ln \dd(b)^q,
            \\ r H(Y|X)
            & \le \sum_{a,b,c} Q(a, b, c) \ln \dd(c)^r.
        \end{align*}
        The sum of the left-hand sides of the last four inequalities is
        \[ (p + r + 1) H(Y|X) + I(X; Y) + (q + 2) H(X|Y). \]
        Applying Jensen's inequality to the sum of the right-hand sides
        to move the logarithm outside
        \begin{align*}
            \kern2em&\kern-2em
            \sum_{a,b,c} Q(a, b, c)
                \ln \frac{\dd(a)^p \dd(b)^q \dd(c)^r}{Q(a, b, c)}
            \\ & \le \ln \sum_{a,b,c} Q(a, b, c)
                \frac{\dd(a)^p \dd(b)^q \dd(c)^r}{Q(a, b, c)}
            \\ & = \ln \sum_{a,b,c} \dd(a)^p \dd(b)^q \dd(c)^r
                = \ln \VVV pqr(R).
        \end{align*}
        This proves the first inequality the theorem claims.

        For $\NNN pqrs(R)$, define a random path
        $(X_0, Y_0, X_1, Y_1)$ by
        \[ Q_\mathsf N (a, b, c, d) \coloneqq Q(a, b, c) P_{Y|X}(d|c) \]
        and similar arguments apply.
        For $\WWW pqrst(R)$, define a random path
        $(X_0, Y_0, X_1, Y_1, X_2)$ by
        \[
            Q_\mathsf W (a, b, c, d, e) \coloneqq
            Q_\mathsf N (a, b, c, d) P_{X|Y}(e|d)
        \]
        and similar arguments apply.
    \end{IEEEproof}

    We next generalize Proposition \ref{pro:holder}.

    \begin{proposition}                                    \label{pro:convex}
        Let $0 < w < 1$ and $\bar w \coloneqq 1 - w$.
        Let $p, q, r, s, t, p'', q'', r'', s'', t'' \ge 0$ be real numbers.
        Let $p'$ be the weighted average $wp + \ww p''$,
        and define $q'$, $r'$, $s'$, and $t'$ similarly.
        Then the following hold
        \begin{align*}
            \VVV pqr(R)^w \cdot \VVV{p''}{q''}{r''}(R)^\ww
            &\ge \VVV{p'}{q'}{r'}(R),
            \\
            \NNN pqrs(R)^w \cdot \NNN{p''}{q''}{r''}{s''}(R)^\ww
            &\ge \NNN{p'}{q'}{r'}{s'}(R),
            \\
            \WWW pqrst(R)^w \cdot \WWW{p''}{q''}{r''}{s''}{t''}(R)^\ww
            &\ge \WWW{p'}{q'}{r'}{s'}{t'}(R).
        \end{align*}
        Consequently,
        $\ln \VVV pqr(R)$, $\ln \NNN pqrs(R)$, and $\ln \WWW pqrst(R)$
        are convex in their parameters.
    \end{proposition}

    \begin{IEEEproof}
        We only prove the first inequality; the other two are similar.
        Apply H\"older's inequality in the form
        \[
            \Bigl( \sum_i u_i \Bigr)^w
            \Bigl( \sum_i v_i \Bigr)^\ww
            \ge \sum_i u_i^w v_i^\ww
        \]
        to the index set $\{(a,b,c): R(a,b)\land R(c,b)\}$ with
        \[
            u_{a,b,c} \coloneqq \dd(a)^p \dd(b)^q \dd(c)^r,
            \quad
            v_{a,b,c} \coloneqq \dd(a)^{p''} \dd(b)^{q''} \dd(c)^{r''}.
        \]
        This yields
        \begin{align*}
            \kern2em&\kern-2em
            \VVV pqr(R)^w \cdot \VVV{p''}{q''}{r''}(R)^\ww
            \\ & = \Bigl( \sum \dd(a)^p \dd(b)^q \dd(c)^r \Bigr)^w
                \Bigl( \sum \dd(a)^{p''} \dd(b)^{q''} \dd(c)^{r''} \Bigr)^\ww
            \\ & \ge \sum
                \Bigl( \dd(a)^p \dd(b)^q \dd(c)^r \Bigr)^w
                \Bigl( \dd(a)^{p''} \dd(b)^{q''} \dd(c)^{r''} \Bigr)^\ww
            \\ & = \sum
                \dd(a)^{wp+\ww p''}
                \dd(b)^{wq+\ww q''}
                \dd(c)^{wr+\ww r''},
        \end{align*}
        which is exactly $\VVV{p'}{q'}{r'}(R)$.
        This proves the first inequality the proposition claims.
        For the other two inequalities, similar arguments apply.
    \end{IEEEproof}

    To conclude, caterpillar moments admit the same
    continuous parametrization as stars and bi-stars.
    In practical implementations one may sample a finite grid of
    parameters, which can be seen as a numerical approximation
    to a continuous convex optimization problem.

\section{Supplemental Plots}

    In the following, we plot the upper bounds
    versus the actual number of graph homomorphisms for each dataset.
    That is, Figure~\ref{fig:sim} is the geometric mean of these plots.
    The name of each dataset is in the lower-right corner of the plot.
    From these plots one can see that the advantage of
    our new methods shows up in most datasets.

\smallskip

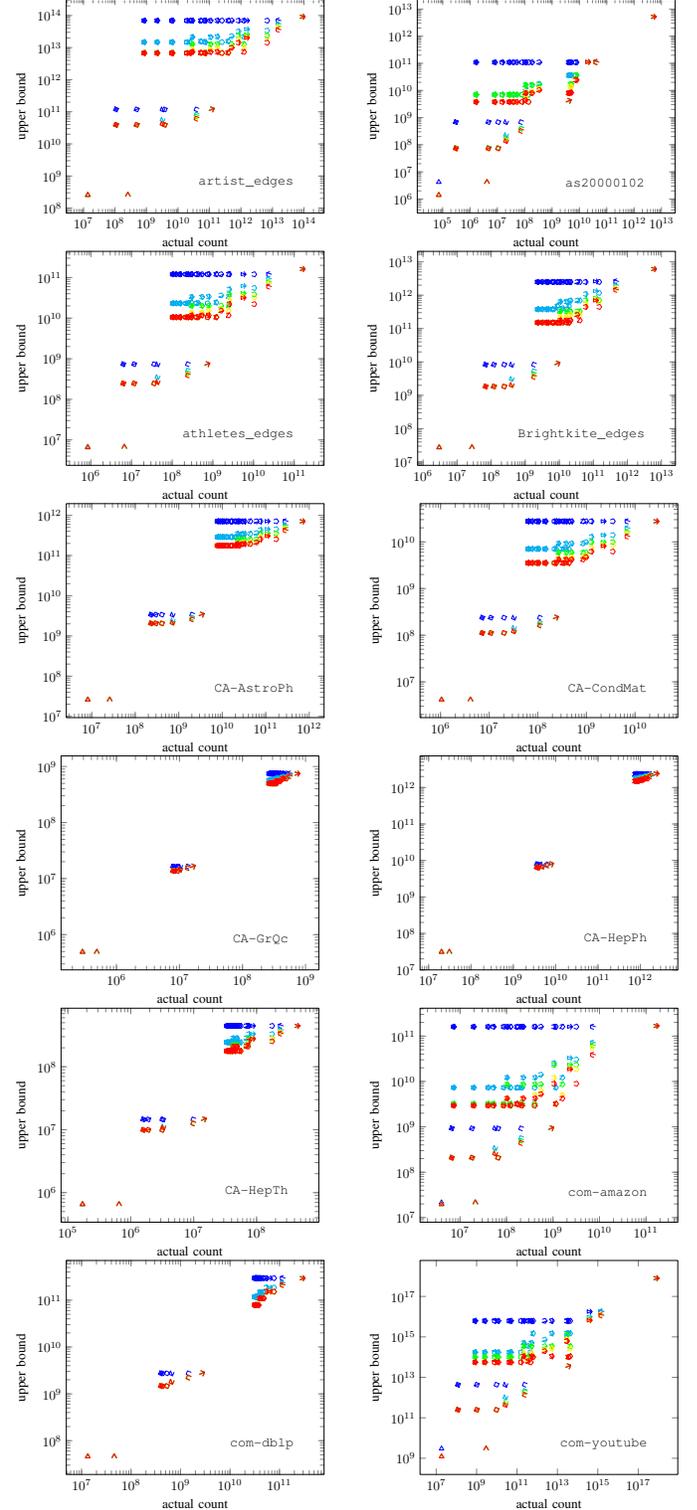
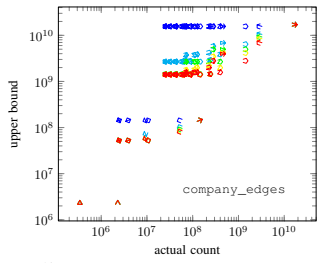
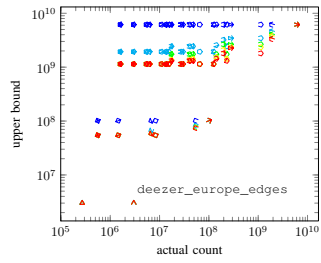
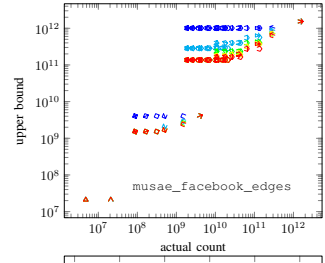
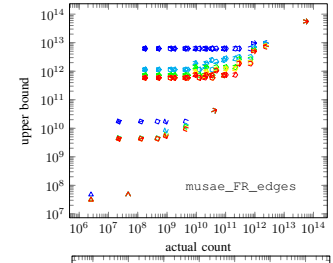
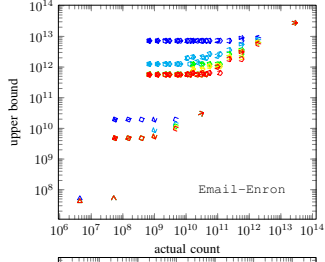
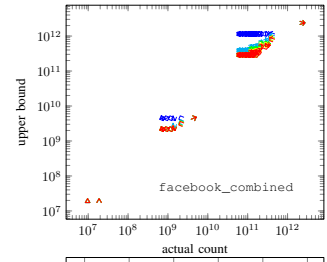
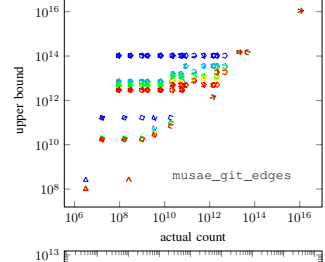
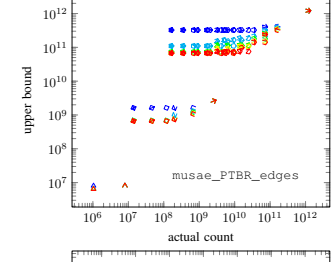
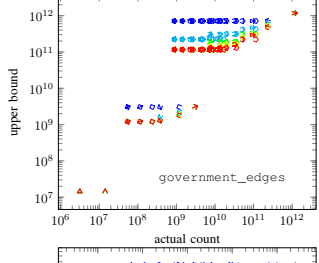
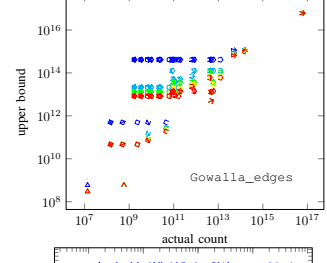
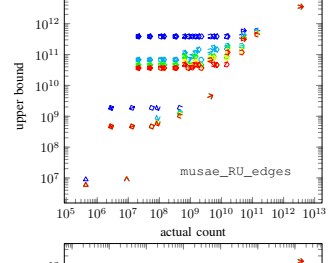
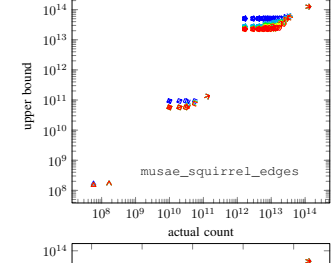
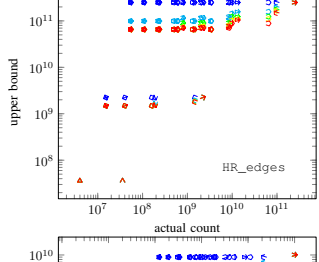
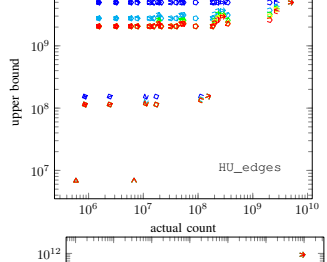
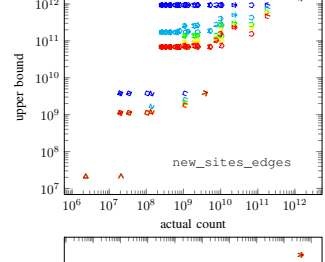
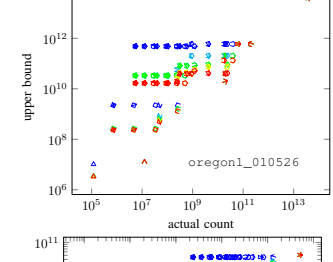
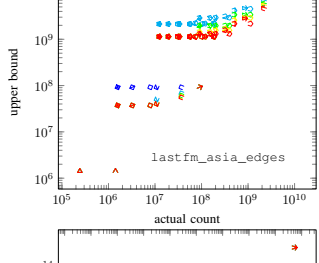
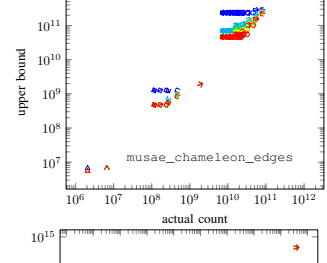
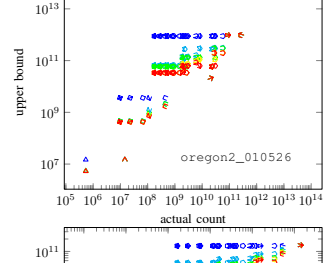
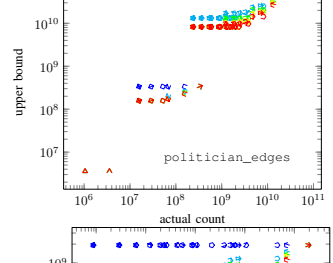
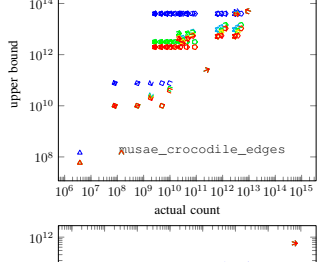
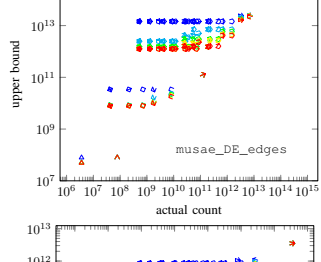
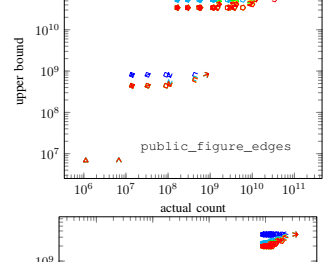
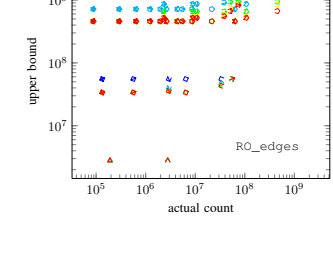
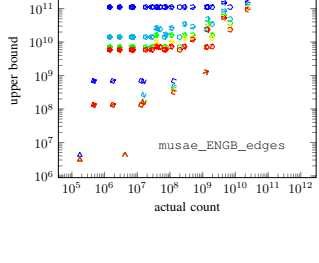
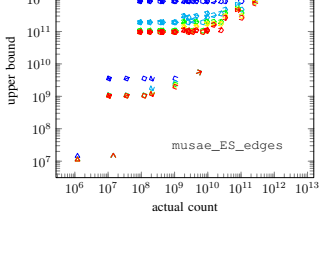
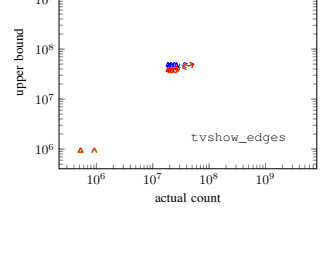
\def\invokefilename#1{%
    \begin{tikzpicture} [scale=0.5]
        \begin{axis} [
            xmode=log, xlabel={actual count},
            ymode=log, ylabel={upper bound},
        ]
            \draw (rel axis cs: 0.9, 0.1) node [above left] {\texttt{#1}};
            \addplot [graph marker=b] table [x=true, y=star, meta=shape]
                {fig54/#1.csv};
            \addplot [graph marker=c] table [x=true, y=bistar, meta=shape]
                {fig54/#1.csv};
            \addplot [graph marker=g] table [x=true, y=vvv, meta=shape]
                {fig54/#1.csv};
            \addplot [graph marker=y] table [x=true, y=nnn, meta=shape]
                {fig54/#1.csv};
            \addplot [graph marker=r] table [x=true, y=www, meta=shape]
                {fig54/#1.csv};
        \end{axis}
    \end{tikzpicture}%
    \hfill
}
\catcode`\_=11
\noindent
\invokefilename{artist_edges}%
\invokefilename{as20000102}%
\invokefilename{athletes_edges}%
\invokefilename{Brightkite_edges}%
\invokefilename{CA-AstroPh}%
\invokefilename{CA-CondMat}%
\invokefilename{CA-GrQc}%
\invokefilename{CA-HepPh}%
\invokefilename{CA-HepTh}%
\invokefilename{com-amazon}%
\invokefilename{com-dblp}%
\invokefilename{com-youtube}%
\invokefilename{company_edges}%
\invokefilename{deezer_europe_edges}%
\invokefilename{Email-Enron}%
\invokefilename{facebook_combined}%
\invokefilename{government_edges}%
\invokefilename{Gowalla_edges}%
\invokefilename{HR_edges}%
\invokefilename{HU_edges}%
\invokefilename{lastfm_asia_edges}%
\invokefilename{musae_chameleon_edges}%
\invokefilename{musae_crocodile_edges}%
\invokefilename{musae_DE_edges}%
\invokefilename{musae_ENGB_edges}%
\invokefilename{musae_ES_edges}%
\invokefilename{musae_facebook_edges}%
\invokefilename{musae_FR_edges}%
\invokefilename{musae_git_edges}%
\invokefilename{musae_PTBR_edges}%
\invokefilename{musae_RU_edges}%
\invokefilename{musae_squirrel_edges}%
\invokefilename{new_sites_edges}%
\invokefilename{oregon1_010526}%
\invokefilename{oregon2_010526}%
\invokefilename{politician_edges}%
\invokefilename{public_figure_edges}%
\invokefilename{RO_edges}%
\invokefilename{tvshow_edges}%
\unskip

\bigskip

    In the following,
    we plot the relative errors of the star-based bounds
    versus the relative errors of our best bounds for each dataset.
    The name of each dataset is in the lower-right corner of the plot.
    The geometric mean of these plots follows.
    From these plots one can see that the advantage of our new methods
    can be characterized by a power law with exponent about $3/5$.

\smallskip

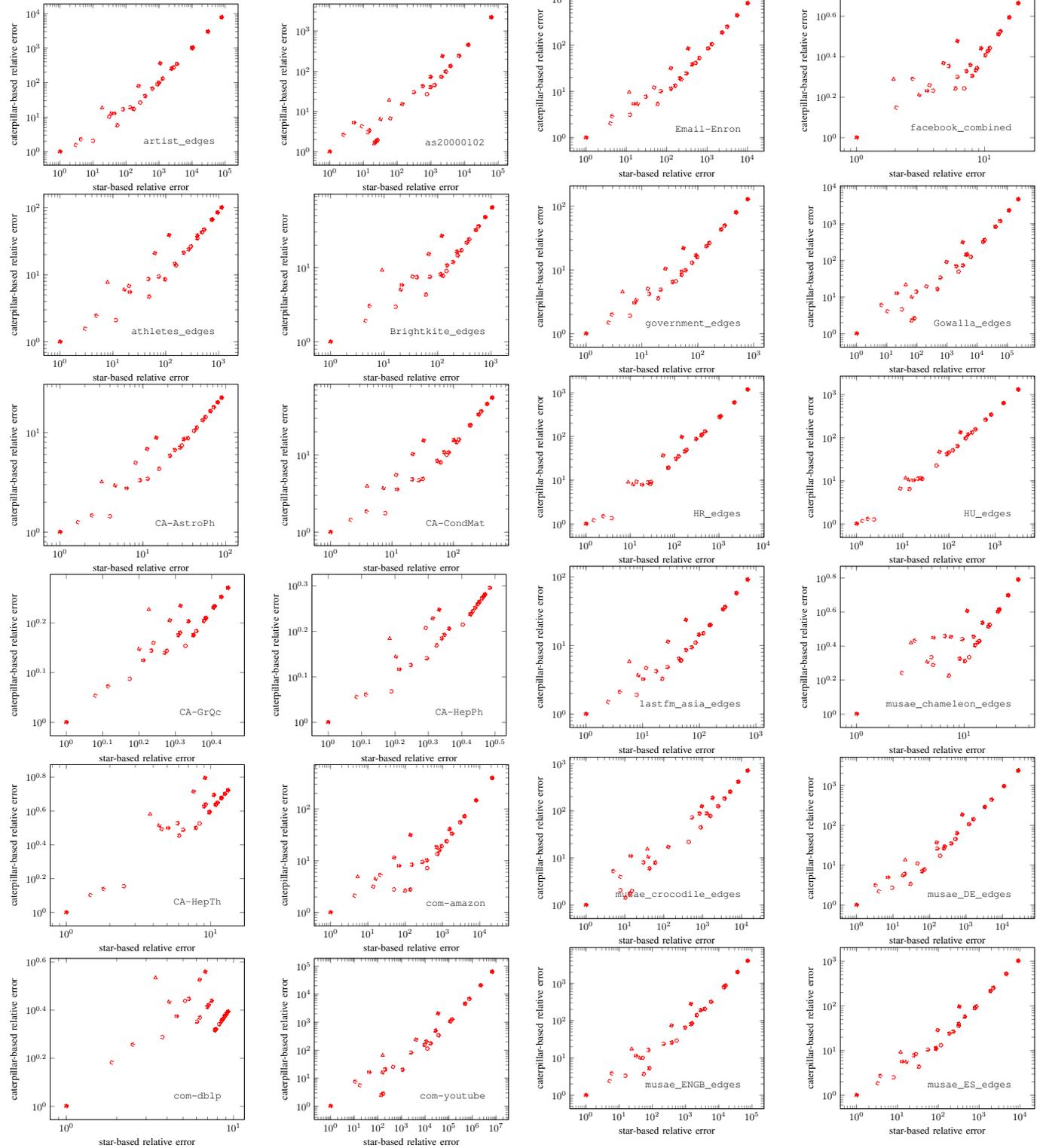
\def\invokefilename#1{%
    \begin{tikzpicture} [scale=0.5]
        \begin{axis} [
            xmode=log, xlabel={star-based relative error},
            ymode=log, ylabel={caterpillar-based relative error},
        ]
            \draw (rel axis cs: 0.9, 0.1) node [above left] {\texttt{#1}};
            \addplot [graph marker=r] table [x=s/t, y=w/t, meta=shape]
                {fig54/#1.csv};
        \end{axis}
    \end{tikzpicture}%
    \hfill
}
\catcode`\_=11
\noindent
\invokefilename{artist_edges}%
\invokefilename{as20000102}%
\invokefilename{athletes_edges}%
\invokefilename{Brightkite_edges}%
\invokefilename{CA-AstroPh}%
\invokefilename{CA-CondMat}%
\invokefilename{CA-GrQc}%
\invokefilename{CA-HepPh}%
\invokefilename{CA-HepTh}%
\invokefilename{com-amazon}%
\invokefilename{com-dblp}%
\invokefilename{com-youtube}%
\invokefilename{company_edges}%
\invokefilename{deezer_europe_edges}%
\invokefilename{Email-Enron}%
\invokefilename{facebook_combined}%
\invokefilename{government_edges}%
\invokefilename{Gowalla_edges}%
\invokefilename{HR_edges}%
\invokefilename{HU_edges}%
\invokefilename{lastfm_asia_edges}%
\invokefilename{musae_chameleon_edges}%
\invokefilename{musae_crocodile_edges}%
\invokefilename{musae_DE_edges}%
\invokefilename{musae_ENGB_edges}%
\invokefilename{musae_ES_edges}%
\invokefilename{musae_facebook_edges}%
\invokefilename{musae_FR_edges}%
\invokefilename{musae_git_edges}%
\invokefilename{musae_PTBR_edges}%
\invokefilename{musae_RU_edges}%
\invokefilename{musae_squirrel_edges}%
\invokefilename{new_sites_edges}%
\invokefilename{oregon1_010526}%
\invokefilename{oregon2_010526}%
\invokefilename{politician_edges}%
\invokefilename{public_figure_edges}%
\invokefilename{RO_edges}%
\invokefilename{tvshow_edges}%
\unskip

\bigskip

    The following is the geometric mean of the plots above.
    One can run log-log regression on this plot and
    find the exponent of the power law to be about $3/5$.
    (More precisely, the trend line is $\log y = 0.6331 \log x$.)
    The R-square of the regression is about $0.9829$.

\smallskip

\begin{tikzpicture}
    \begin{axis} [
        xmode=log, xlabel={star-based relative error (averaged)},
        ymode=log, ylabel={caterpillar-based relative error (averaged)},
    ]
        \addplot [graph marker=r] table [x=s/t, y=w/t, meta=shape]
            {fig54/_average.csv};
    \end{axis}
\end{tikzpicture}%

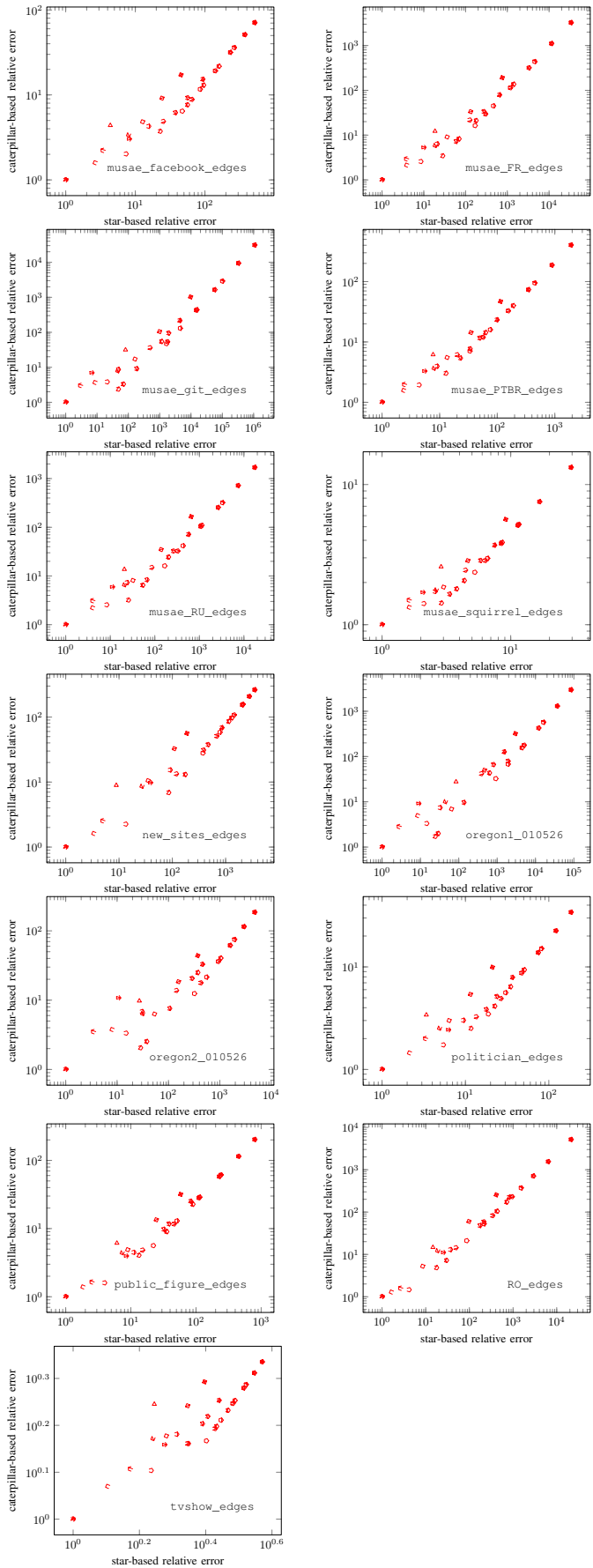
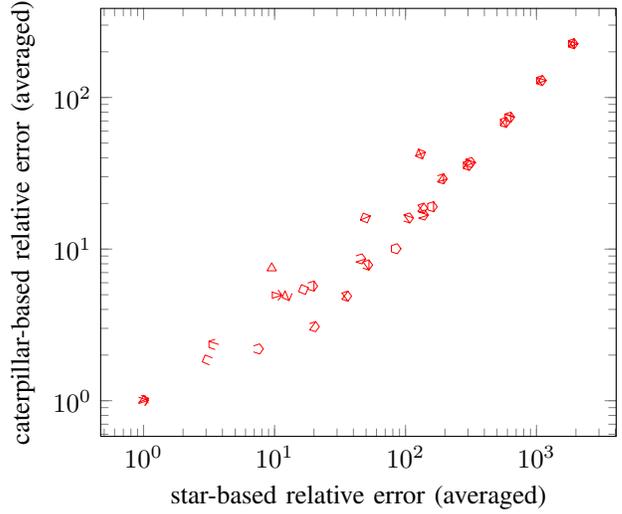
\end{document}